\documentclass[11pt]{article}

\usepackage{authblk}
\usepackage[colorlinks,linkcolor=blue,anchorcolor=red,citecolor=red,pdfstartview=FitH]{hyperref}
\usepackage{graphics,graphicx}
\usepackage{mathtools}
\usepackage{amsmath,amssymb,amsthm,accents,amsfonts}
\usepackage{color}
\usepackage{cite}
\usepackage{dsfont}
\newtheorem{theorem}{Theorem}[section]

\numberwithin{equation}{section}
\allowdisplaybreaks

\begin{document}

\title{\bf {On Darboux transformations for the derivative nonlinear Schr\"{o}dinger equation}}
\author[1]{\bf  Jonathan J.C. Nimmo}
\author[1,2]{\bf Halis Yilmaz}
\affil[1]{School of Mathematics and Statistics, University of Glasgow, Glasgow G12 8QW, UK}
\affil[2]{Department of Mathematics, University of Dicle, 21280 Diyarbakir, Turkey}
\date{}

\maketitle

\begin{abstract}
We consider Darboux transformations for the derivative nonlinear Schr\"{o}dinger equation.
 A new theorem for Darboux transformations of operators with no derivative term are presented and proved.
 The solution is expressed in quasideterminant forms. Additionally, the parabolic and soliton solutions
 of the derivative nonlinear Schr\"{o}dinger equation  are given as explicit examples.
\end{abstract}

\quad{\text{\it{Keywords:}} Derivative nonlinear Schr\"{o}dinger equation; Darboux transformation;  Quasideterminants.}

\quad{\text{\it{2010 Mathematics Subject Classification:}} 35C08, 35Q55, 37K10, 37K35}


\section{Introduction}

The derivative nonlinear Schr\"{o}dinger equation (DNLS), also called the Kaup-Newell (KN) equation,
\begin{eqnarray}\label{DNLS}
  iq_t+q_{xx}=i(|q|^2 q)_x
\end{eqnarray}
is completely integrable and is an important model in mathematical physics, especially in space plasma physics and nonlinear optics \cite{Agr,And,Ich,Mj,Rud,Tz}.
Kaup and Newell \cite{KN} solved the initial value problem for the DNLS equation using the inverse scattering method.

Darboux transformations are an important tool for studying the solutions of integrable systems.
They provide a universal approach that will bring together and extend a number of disparate results connected with
the nonlinear Schr\"{o}dinger (NLS) equation and its cousin the derivative NLS equation (DNLS).
In recent years, there has been some interest in solutions of the DNLS equation obtained by means of \textit{Darboux-like} transformations \cite{Steudel,Xia,Xu,Zang}. These solutions are often written in terms of determinants with a complicated structure.
Here, under a gauge transformation, a one-step Darboux transformation of the KN system (\ref{DNLS1}--\ref{DNLS2}) is
constructed by finding a $2 \times 2$ trial matrix so that the KN spectral problem (Lax pair) \cite{KN} is covariant.
Then, the determinant representations of $n$-fold Darboux transformation are obtained by stating and proving of sequence of theorems.
These determinants are expressed in terms of  solutions (eigenfunctions) of the linear partial differential equations,
where the equations (\ref{DNLS1}) and (\ref{DNLS2}) are the integrable condition of this linear KN system.
The important point to note here is that Steudel \cite{Steudel} has established a general formulae of the solution of the KN system
in terms of \textit{Vandermonde-type determinants} \cite{Steudel.M.N.}.
He used solutions of Riccati equations, which are replaced by solutions of the linear KN system, in order to construct solutions of the DNLS equation.
Steudel introduced his \textit{seahorse function} to write down general solutions of Riccati equations in terms of this auxiliary function.

On the other hand, in this present paper, we present a systematic approach to the construction of solutions of \eqref{DNLS}
by means of a standard Darboux transformation and written in terms of quasideterminants.
Quasideterminants have various nice properties which play important roles in constructing exact solutions of integrable systems.
The reader is referred to the original papers for a detailed description of quasideterminants \cite{Gelfand91,Gelfand05}
and their applications in integrable systems \cite{CJ,CHNimmo,CS,Hassan,LiJon}.

For the sake of clarity we emphasize that the strategy we employ here is based on Darboux's \cite{Darboux} and Matveev's \cite{Matveev,MS}  original ideas.
Therefore, our approach should be considered on its own merits.

This paper is organised as follows. In Section \ref{DTs}, we state and prove a theorem for Darboux transformations of operators with no derivative term. This has a similar structure to the standard theorem \cite{MS} for Darboux transformation of general operators.

In Sections \ref{DT1} and \ref{DT2}, we show how the quasideterminant solutions of the DNLS equation arise naturally from the Darboux transformation. Here, the quasideterminants are written in terms of solutions of Riccati sytems which arise from linear eigenvalue problems.

In Section \ref{PS}, parametric and soliton solutions of the DNLS equation are given for both zero and non-zero seed solutions.

\section{Derivative Nonlinear Schr\"{o}dinger equations}

Let us consider the coupled DNLS equations
\begin{eqnarray}
 iq_t+q_{xx}-i(q^2r)_x&=&0, \label{DNLS1}\\
 ir_t-r_{xx}-i(r^2q)_x&=&0, \label{DNLS2}
\end{eqnarray}
where $q=q(x,t)$ and $r=r(x,t)$ are complex-valued functions.
Equations (\ref{DNLS1}) and (\ref{DNLS2}) reduce to the DNLS equation \eqref{DNLS} when $r=q^*$,
where $q^*$ denotes the complex conjugation of $q$.

The Lax pair for the coupled DNLS equations (\ref{DNLS1})-(\ref{DNLS2}) is given by
\begin{eqnarray}
 L&=&\partial_x+J\lambda^2-R\lambda\label{LaxL}\\
 M&=&\partial_t+2J\lambda^4-2R\lambda^3+qrJ\lambda^2+U\lambda,\label{LaxM}
\end{eqnarray}
where $J$, $R$ and $U$ are the $2\times 2$ matrices
\begin{eqnarray}
  J={\left(\begin{array}{cc}
          i & 0\\ 0 & -i
  \end{array}\right)},\hspace{0.3cm}
   R={\left(\begin{array}{cc}
          0& q\\ r & 0
  \end{array}\right)}\hspace{0.3cm}\text{and}\hspace{0.3cm}
   U={\left(\begin{array}{cc}
          0& -iq_x-rq^2\\ ir_x-r^2q & 0
  \end{array}\right)}.
\end{eqnarray}
Here $\lambda$ is an arbitrary complex number, the eigenvalue (or spectral parameter).

\section{Darboux Theorems and Dimensional Reductions}\label{DTs}

\subsection{General Darboux theorems}
\begin{theorem}[\cite{Darboux,Matveev,MS}] Consider the linear operator
\begin{equation}\label{L}
 L=\partial_x+\sum_{i=0}^nu_i\partial_y^i
\end{equation}
where $u_i\in R$, where $R$ is a ring, in general non-commutative. Let $G=\theta\partial_y\theta^{-1}$, where $\theta=\theta(x,y)$ is an invertible eigenfunction of $L$, so that $L(\theta)=0$.
Then
\begin{equation}
 \tilde{L}=GLG^{-1}
\end{equation}
has the same form as $L$:
\begin{equation}
 \tilde{L}=\partial_x+\sum_{i=0}^n\tilde{u}_i\partial_y^i
\end{equation}
If $\phi$ is any eigenfunction of $L$ then
\begin{equation}
 \tilde{\phi}=\phi_x-\theta_y\theta^{-1}\phi
\end{equation}
is an eigenfunction of $\tilde{L}$. In other words, if $L(\phi)=0$ then $\tilde{L}(\tilde{\phi})=0$ where $\tilde{\phi}=G(\phi)$.
\end{theorem}

This Darboux transformation does not, however, preserve the form of $L$ when $u_0=0$. That is, for $L$ with $u_0=0$, $\tilde{u}_0\neq 0$ in general. In the scalar case \cite{O.R} and matrix case \cite{JonKP}, it is shown that the operator \eqref{L} with $u_0=0$ is invariant under the Darboux transformation
\begin{eqnarray}\label{DTJon}
 G=\left[\left(\theta^{-1}\right)_y\right]^{-1}\partial_y\theta^{-1}.
\end{eqnarray}
Unfortunately, this transformation acts trivially for the DNLS and we need to consider a slight generalization.
\begin{theorem}
Consider the linear operator
\begin{equation}\label{L1}
 L=\partial_x+\sum_{i=1}^nu_i\partial_y^i
\end{equation}
where $u_i\in R$. Let
\begin{eqnarray}\label{GDNLS}
  G=\sigma\left[\left(\theta^{-1}\right)_y\right]^{-1}\partial_y\theta^{-1},
\end{eqnarray}
where $\theta=\theta(x,y)$ is an invertible eigenfunction of $L$ and $\sigma\in R$ is invertible and independent of $x$ and $y$.
Then
\begin{equation}
 \tilde{L}=GLG^{-1}
\end{equation}
has the same form as $L$:
\begin{equation}
 \tilde{L}=\partial_x+\sum_{i=1}^n\tilde{u}_i\partial_y^i
\end{equation}
If $\phi$ is any eigenfunction of $L$ then
\begin{equation}
 \tilde{\phi}=\sigma\phi-\sigma\theta\left(\theta_y\right)^{-1}\phi_y
\end{equation}
is an eigenfunction of $\tilde{L}$. In other words if  $L(\phi)=0$ then $\tilde{L}(\tilde{\phi})=0$ where $\tilde{\phi}=G(\phi)$.
\end{theorem}

\begin{proof}
The case $\sigma=I$ is proved in \cite{JonKP}. For the case of general $\sigma$, it is sufficient to observe that under the transformation $G\to\sigma G$, $\tilde L\to \sigma\tilde L\sigma^{-1}$ and, since $\sigma$ is constant, $\tilde u_i\to\sigma\tilde u_i\sigma^{-1}$. Thus the structure of $L$ is preserved by $G$ given in \eqref{GDNLS}.
\end{proof}

\subsection{Dimensional reduction of Darboux transformation}
Here, we describe a reduction of the Darboux transformation from $(2+1)$ to $(1+1)$ dimensions. We choose to  eliminate the $y$-dependence by employing a `separation of variables' technique. The reader is referred to the paper \cite{NGO} for a more detailed treatment. We make the ansatz
\begin{eqnarray}
 \phi &=&\phi^r(x,t)e^{\lambda y},\\
 \theta &=&\theta^r(x,t)e^{\Lambda y},
\end{eqnarray}
where $\lambda$ is a constant scalar and $\Lambda$ an $N \times N$ constant matrix and
the superscript $r$ labels reduced functions, independent of $y$.
Hence in the dimensional reduction we obtain $\partial_y^{i}\left(\phi\right)=\lambda^i\phi$ and
$\partial_y^{i}\left(\theta \right)=\theta \Lambda^i$ and so the operator $L$ and Darboux transformation $G$ become
 \begin{eqnarray}\label{redL}
  L^r&=&\partial_x+\sum_{i=1}^n u_i\lambda^i,\\
  G^r&=&\sigma-\sigma\theta^r\Lambda^{-1}(\theta^{r})^{-1}\lambda,
 \end{eqnarray}
where $\theta^r$ is a matrix eigenfunction of $L^r$ such that $L^r\left(\theta^r\right)=0$, with $\lambda$ replaced by matrix $\Lambda$, that is,
\begin{equation}
	\theta^r_x+\sum_{i=1}^n u_i\theta^r\Lambda^i=0.
\end{equation}
Below we omit the superscript $r$ for a simpler notation.

\subsection{Iteration of reduced Darboux Transformations}\label{DT1}

In this section we shall consider iteration of the Darboux transformation and find closed form expressions for these in terms of quasideterminants. The reader is referred to \cite{Gelfand91,Gelfand05} for an explanation of the quasideterminant notation.

Let $L$ be an operator, form invariant under the reduced Darboux transformation $G_{\theta}=\sigma-\sigma\theta\Lambda^{-1}\theta^{-1}\lambda$ discussed above.

Let $\phi=\phi(x,t)$ be a general eigenfunction of $L$ such that $L(\phi)=0$. Then
\begin{eqnarray*}
\tilde{\phi}&=&G_\theta\left(\phi\right)\\
             &=&\sigma \left(\phi-\theta\Lambda^{-1}\theta^{-1}\lambda\phi\right)\\
             &=&\sigma\left|\begin{array}{cc} \theta & \boxed{\phi}\\ \theta\Lambda & \lambda\phi\end{array}\right|
\end{eqnarray*}
is an eigenfunction of $\tilde{L}=G_\theta L G_\theta^{-1}$ so that
$\tilde{L}(\tilde{\phi})=\lambda\tilde{\phi}$. Let $\theta_i$ for $i=1,\ldots,n,$ be a particular set
of invertible eigenfunctions of $L$ so that $L(\theta_i)=0$ for $\lambda=\Lambda_i$, and introduce the notation $\Theta=(\theta_1,\ldots,\theta_n)$. To apply the Darboux transformation a second time, let $\theta_{[1]}=\theta_1$ and $\phi_{[1]}=\phi$ be a general eigenfunction of $L_{[1]}=L$.
Then
$\phi_{[2]}=G_{\theta_{[1]}}\left(\phi_{[1]}\right)$ and $\theta_{[2]}=\phi_{[2]}|_{\phi\rightarrow \theta_2}$ are eigenfunctions for $L_{[2]}=G_{\theta_{[1]}}L_{[1]}G_{\theta_{[1]}}^{-1}$.

In general, for $n\geq 1$, we define the $n$th Darboux transform of $\phi$ by
\begin{equation}
 \phi_{[n+1]}=\sigma\left(\phi_{[n]}-\theta_{[n]}\Lambda_n^{-1}\theta_{[n]}^{-1}\lambda \phi_{[n]}\right),
\end{equation}
in which
\begin{equation*}
 \theta_{[k]}=\phi_{[k]}|_{\phi\rightarrow\theta_k}~.
\end{equation*}
For example,
\begin{eqnarray*}
\phi_{[2]}&=&\sigma \left(\phi-\theta_1\Lambda_1^{-1}\theta_1^{-1}\lambda\phi\right)
         =\sigma\left|\begin{array}{cc} \theta_1 & \boxed{\phi}\\ \theta_1\Lambda_1&\lambda\phi\end{array}\right|,\\
\phi_{[3]}&=&\sigma \left(\phi_{[2]}-\theta_{[2]}\Lambda_2^{-1}\theta_{[2]}^{-1}\lambda\phi_{[2]}\right)\\
          &=&\sigma^2\left|\begin{array}{ccc} \theta_1 & \theta_2 & \boxed{\phi}\\
            \theta_1\Lambda_1 & \theta_2 \Lambda_2 & \lambda \phi\\
            \theta_1\Lambda_1^2&\theta_2\Lambda_2^2 & \lambda^2 \phi\end{array}\right|.
\end{eqnarray*}
After $n$ iterations, we get
\begin{eqnarray}
\phi_{[n+1]}= \sigma^n \left|\begin{array}{ccccc}
      \theta_1 & \theta_2 \hspace{0.2cm} \ldots \hspace{0.2cm} \theta_n & \boxed{\phi}\\
      \theta_1\Lambda_1 & \theta_2\Lambda_2 \ldots \theta_n\Lambda_n & \lambda\phi\\
      \theta_1\Lambda_1^2 & \theta_2\Lambda_2^2 \ldots \theta_n\Lambda_n^2 & \lambda^2 \phi\\
      \vdots & \vdots \hspace{0.3cm} \ldots \hspace{0.3cm} \vdots & \vdots \\
      \theta_1\Lambda_1^n & \theta_2\Lambda_2^n \ldots \theta_n\Lambda_n^n & \lambda^n \phi\\
     \end{array}
 \right|.
\end{eqnarray}

\section{Constructing Solutions for DNLS Equation}\label{DT2}
In this section we determine the specific effect of the Darboux transformation $G=\sigma-\sigma\theta\Lambda^{-1}\theta^{-1}\lambda$
on the $2\times2$ Lax operators $L,M$ given by \eqref{LaxL},\eqref{LaxM}. Here $\theta$ is a eigenfunction satisfying $L(\theta)=M(\theta)=0$ with $2\times2$ matrix eigenvalue $\Lambda$. From $\tilde LG=GL$ we obtain the three conditions
\begin{align}
	\label{C1}
	[J,\sigma\theta\Lambda^{-1}\theta^{-1}]&=0\\
	\label{C2}
	\tilde R\sigma\theta\Lambda^{-1}\theta^{-1}&=\sigma\theta\Lambda^{-1}\theta^{-1}R+[\sigma,J]\\
	\label{C3}
	\tilde R\sigma&=\sigma R-(\sigma\theta\Lambda^{-1}\theta^{-1})_x.
\end{align}
From \eqref{C1}, we see that $\sigma\theta\Lambda^{-1}\theta^{-1}$ must be a diagonal matrix and then from \eqref{C2} that $[\sigma,J]$ and hence $\sigma$ must be off-diagonal. Guided by this, we choose
\begin{eqnarray}
 \Lambda=\left(\begin{array}{cc} 1 & 0\\ 0 & -1\end{array}\right)\lambda ,\hspace{1cm} \sigma=\left(\begin{array}{cc} 0 & 1\\ 1 & 0\end{array}\right).
\end{eqnarray}
Finally, comparison of \eqref{C2} and \eqref{C3} leads to the requirement that the matrix $\theta$ has the structure
\begin{eqnarray}
 \theta=\left(\begin{array}{cc} \theta_{11} & \theta_{12}\\ f\theta_{11} & -f\theta_{12}\end{array}\right),
\end{eqnarray}
and in turn the linear equations for $\theta$ impose conditions of $f$, namely the Riccati equations
\begin{eqnarray}
 f_x+\lambda qf^2-2\lambda^2if-\lambda r=0 \label{riccati1}\\
 f_t+\lambda\left(iq_x+rq^2+2\lambda^2q\right)f^2-2i\lambda^2\left(2\lambda^2+qr\right)f-\lambda\left(qr^2+2\lambda^2r-ir_x\right)=0 \label{riccati2}
\end{eqnarray}
for given $q(x,t), r(x,t)$ solutions in (\ref{DNLS1}-\ref{DNLS2}) and $\lambda$ is a constant scalar.

In summary, the Darboux transformation is
\begin{eqnarray}\label{tilR}
 \sigma \tilde{R} \sigma =R-\left(\theta \Lambda^{-1} \theta^{-1}\right)_x
\end{eqnarray}
which can be written in a quasideterminant structure as
\begin{eqnarray}
 \sigma \tilde{R} \sigma =R+\left|\begin{array}{cc} \theta & \boxed{0_2}\\ \theta\Lambda & I_2\end{array}\right|_x,
\end{eqnarray}
\newline
We rewrite \eqref{tilR} as
\begin{eqnarray}
 \sigma R_{[2]} \sigma =R_{[1]}-\left(\theta_{[1]} \Lambda_1^{-1} \theta_{[1]}^{-1}\right)_x
\end{eqnarray}
where $R_{[1]}=R$, $R_{[2]}=\tilde{R}$, $\theta_{[1]}=\theta_1=\theta$, $f_1=f$,  $\Lambda_1=\Lambda$, $\lambda=\lambda_1$. Then after repeated $n$ Darboux transformations, we have
\begin{eqnarray}
 \sigma R_{[n+1]} \sigma =R_{[n]}-\left(\theta_{[n]} \Lambda_n^{-1} \theta_{[n]}^{-1}\right)_x
\end{eqnarray}
which can be written in quasideterminant form as
\begin{eqnarray}\label{Rn1}
\sigma^nR_{[n+1]}\sigma^n&=&R+\left|\begin{array}{ccccc}
      \theta_1 & \theta_2 \hspace{0.2cm} \ldots \hspace{0.2cm} \theta_n & \boxed{0_2}\\
      \theta_1\Lambda_1 & \theta_2\Lambda_2 \ldots \theta_n\Lambda_n & I_2\\
      \theta_1\Lambda_1^2 & \theta_2\Lambda_2^2 \ldots \theta_n\Lambda_n^2 & 0_2\\
      \vdots & \vdots \hspace{0.3cm} \ldots \hspace{0.3cm} \vdots & \vdots \\
      \theta_1\Lambda_1^n & \theta_2\Lambda_2^n \ldots \theta_n\Lambda_n^n & 0_2\\
     \end{array}
 \right|_x,
\end{eqnarray}
where
\begin{eqnarray}
 \theta_i\Lambda_i^k=\left(\begin{array}{cc} \phi_{2i-1} & (-1)^k\phi_{2i}\\ f_i\phi_{2i-1} & (-1)^{k+1}\phi_{2i}\end{array}\right)\lambda_i^k,
\end{eqnarray}
where $i=1, \ldots ,n$, $k=0, \ldots, n$ and $f_i$ is a solution of the Riccati equations \eqref{riccati1}-\eqref{riccati2}.

Let
\begin{eqnarray}
 \Theta=\left(\theta_1,\ldots,\theta_n\right)
       =\left(\begin{array}{c}\phi\\\psi \end{array}\right),
\end{eqnarray}
where $\phi=\left(\phi_1, \phi_2 , \ldots, \phi_{2n-1}, \phi_{2n}\right)$ and
$\psi=\left(f_1\phi_1,-f_1\phi_2, \ldots, f_n\phi_{2n-1}, -f_n\phi_{2n}\right)$ denote $1 \times 2n$ row vectors.
Thus, \eqref{Rn1} can be rewritten as
\begin{eqnarray}
  \sigma^n R_{[n+1]}\sigma^n=R+\left|\begin{array}{cc} \Theta & \boxed{0_2}\\ \widehat{\Theta} & E\end{array}\right|_x,
\end{eqnarray}
where $\widehat{\Theta}=\left(\theta_i\Lambda_i^j\right)_{i,j=1,\ldots,n}$ and $E=\left(e_1, e_2\right)$ denote
$2n \times 2n$ and $2n \times 2$ matrices respectively, where $e_i$ represents a column vector with $1$ in the $i^{th}$ row and zeros elsewhere.
Hence, we obtain


\begin{eqnarray}
 \sigma^n R_{[n+1]}\sigma^n=R+\left(\begin{array}{cc}
  \left|\begin{array}{cc} \phi & \boxed{0}\\ \widehat{\Theta} & e_1\end{array}\right|_x &
  \left|\begin{array}{cc} \phi & \boxed{0}\\ \widehat{\Theta} & e_2\end{array}\right|_x\\\\
  \left|\begin{array}{cc} \psi & \boxed{0}\\ \widehat{\Theta} & e_1\end{array}\right|_x &
  \left|\begin{array}{cc} \psi & \boxed{0}\\ \widehat{\Theta} & e_2\end{array}\right|_x \end{array}\right),
\end{eqnarray}
where it can be easily shown that
\begin{eqnarray}
  \begin{array}{cc}
  \left|\begin{array}{cc} \phi & \boxed{0}\\ \widehat{\Theta} & e_1\end{array}\right|=
  \left|\begin{array}{cc} \psi & \boxed{0}\\ \widehat{\Theta} & e_2\end{array}\right|=0\end{array}.
\end{eqnarray}
The pair of $q_{[n+1]}$ and $r_{[n+1]}$ are derived from the above matrix expression with respect to $n$ which is odd $(n=2k-1)$ or even number $(n=2k)$, where $k\in \mathbb{N}$ is a positive integer.

\subsection*{In the case of n odd ($n=2k-1$)}
\begin{eqnarray}
 \begin{array}{cc}
  q_{[n+1]}=r+\left|\begin{array}{cc} \psi & \boxed{0}\\ \widehat{\Theta} & e_1\end{array}\right|_x\end{array},
\end{eqnarray}
\begin{eqnarray}
  \begin{array}{cc}
  r_{[n+1]}=q+\left|\begin{array}{cc} \phi & \boxed{0}\\ \widehat{\Theta} & e_2\end{array}\right|_x\end{array}.
\end{eqnarray}
Thus, we obtain
\begin{eqnarray}\label{qodd}
 \begin{array}{cc}
  q_{[n+1]}=r+\left|\begin{array}{cc} \mathfrak{f} & \boxed{0}\\ \Omega_q & \mathfrak{e_1} \end{array}\right|_x\end{array},
\end{eqnarray}
\begin{eqnarray}\label{rodd}
  \begin{array}{cc}
  r_{[n+1]}=q+\left|\begin{array}{cc} \mathfrak{1} & \boxed{0}\\ \Omega_r & \mathfrak{e_1} \end{array}\right|_x\end{array},
\end{eqnarray}
where
 $\mathfrak{e_1}=(1,0,\ldots,0)^T$, $\mathfrak{1}=(1,1,\ldots,1)$, $\mathfrak{f}=(f_1,f_2,\ldots,f_n)$ and
 \begin{eqnarray}\label{oddmatrix}
 \Omega_q=\left(\begin{array}{ccccc}
      \lambda_1 & \lambda_2 & \ldots & \lambda_n\\
      f_1\lambda_1^2 & f_2\lambda_2^2 & \ldots & f_n\lambda_n^2\\
      \lambda_1^3 & \lambda_2^3 & \ldots & \lambda_n^3\\
      \vdots & \vdots & \vdots & \vdots \\
      f_1\lambda_1^{n-1} & f_2\lambda_2^{n-1} & \ldots & f_n\lambda_n^{n-1}\\
      \lambda_1^n & \lambda_2^n & \ldots & \lambda_n^n\\
     \end{array}
 \right),~
 \Omega_r=\left(\begin{array}{ccccc}
      f_1\lambda_1 & f_2\lambda_2 & \ldots & f_n\lambda_n\\
      \lambda_1^2 & \lambda_2^2 & \ldots & \lambda_n^2\\
      f_1\lambda_1^3 & f_2\lambda_2^3 & \ldots & f_n\lambda_n^3\\
      \vdots & \vdots & \vdots & \vdots \\
      \lambda_1^{n-1} & \lambda_2^{n-1} & \ldots & \lambda_n^{n-1}\\
      f_1\lambda_1^n & f_2\lambda_2^n & \ldots & f_n\lambda_n^n\\
     \end{array}
 \right).
\end{eqnarray}
For $n=1$, we obtain the pair of new solutions for the couple DNLS equations \eqref{DNLS1}-\eqref{DNLS2}
\begin{eqnarray}
 \begin{array}{cc}
  q_{[2]}=r+\left|\begin{array}{cc} f_1 & \boxed{0}\\ \lambda_1 & 1 \end{array}\right|_x\end{array}\nonumber\\
         =r-\frac{1}{\lambda_1}f_{1,x},\hspace{0.95cm}\label{qr1}
 \end{eqnarray}
\begin{eqnarray}
  \begin{array}{cc}
  r_{[2]}=q+\left|\begin{array}{cc} 1 & \boxed{0}\\ f_1\lambda_1 & 1 \end{array}\right|_x\end{array}\nonumber\\
         =q-\frac{1}{\lambda_1}\left(\frac{1}{f_1}\right)_x,\hspace{0.55cm}
\end{eqnarray}
where $f_1$ is a solution of the Riccati equations \eqref{riccati1}-\eqref{riccati2}.
\subsection*{In the case of n even ($n=2k$)}
\begin{eqnarray}
 \begin{array}{cc}
  q_{[n+1]}=q+\left|\begin{array}{cc} \phi & \boxed{0}\\ \widehat{\Theta} & e_2 \end{array}\right|_x\end{array},
\end{eqnarray}
\begin{eqnarray}
  \begin{array}{cc}
  r_{[n+1]}=r+\left|\begin{array}{cc} \psi & \boxed{0}\\ \widehat{\Theta}& e_1\end{array}\right|_x\end{array}.
\end{eqnarray}
Thus, we obtain
\begin{eqnarray}\label{qeven}
 \begin{array}{cc}
  q_{[n+1]}=q+\left|\begin{array}{cc} \mathfrak{1} & \boxed{0}\\ \mho_q & \mathfrak{e_1}\end{array}\right|_x\end{array},
\end{eqnarray}
\begin{eqnarray}\label{reven}
  \begin{array}{cc}
  r_{[n+1]}=r+\left|\begin{array}{cc} \mathfrak{f} & \boxed{0}\\ \mho_r& \mathfrak{e_1} \end{array}\right|_x\end{array},
\end{eqnarray}
where $\mathfrak{e_1}=(1,0,\ldots,0)^T$, $\mathfrak{1}=(1,1,\ldots,1)$, $\mathfrak{f}=(f_1,f_2,\ldots,f_n)$ and
 \begin{eqnarray}\label{evenmatrix}
 \mho_q=\left(\begin{array}{ccccc}
      f_1\lambda_1 & f_2\lambda_2 & \ldots & f_n\lambda_n\\
      \lambda_1^2 & \lambda_2^2 & \ldots & \lambda_n^2\\
      f_1\lambda_1^3 & f_2\lambda_2^3 & \ldots & f_n\lambda_n^3\\
      \vdots & \vdots & \vdots & \vdots \\
      f_1\lambda_1^{n-1} & f_2\lambda_2^{n-1} & \ldots & f_n\lambda_n^{n-1}\\
      \lambda_1^n & \lambda_2^n & \ldots & \lambda_n^n\\
     \end{array}
 \right),~
 \mho_r=\left(\begin{array}{ccccc}
      \lambda_1 & \lambda_2 & \ldots & \lambda_n\\
      f_1\lambda_1^2 & f_2\lambda_2^2 & \ldots & f_n\lambda_n^2\\
      \lambda_1^3 & \lambda_2^3 & \ldots & \lambda_n^3\\
      \vdots & \vdots & \vdots & \vdots \\
      \lambda_1^{n-1} & \lambda_2^{n-1} & \ldots & \lambda_n^{n-1}\\
      f_1\lambda_1^n & f_2\lambda_2^n & \ldots & f_n\lambda_n^n\\
     \end{array}
 \right).
\end{eqnarray}

For $n=2$, we have

\begin{eqnarray}
 \begin{array}{ccc}
  q_{[3]}=q+\left|\begin{array}{ccc} 1 & 1 & \boxed{0}\\ f_1\lambda_1 &f_2 \lambda_2 & 1 \\ \lambda_1^2 & \lambda_2^2 & 0
  \end{array}\right|_x\end{array},
\end{eqnarray}
\begin{eqnarray}
 \begin{array}{ccc}
  r_{[3]}=r+\left|\begin{array}{ccc} f_1 & f_2 & \boxed{0}\\ \lambda_1 & \lambda_2 & 1 \\ f_1\lambda_1^2 & f_2\lambda_2^2 & 0
  \end{array}\right|_x\end{array}.
\end{eqnarray}
Thus, we obtain the pair of new solutions for the couple DNLS equations \eqref{DNLS1}-\eqref{DNLS2}
\begin{eqnarray}
 q_{[3]}&=&q+\frac{\lambda_1^2-\lambda_2^2}{\lambda_1\lambda_2} \left(\frac{1}{\lambda_2f_1-\lambda_1f_2}\right)_x,\label{q3}\\
 r_{[3]}&=&r-\frac{\lambda_1^2-\lambda_2^2}{\lambda_1\lambda_2} \left(\frac{f_1f_2}{\lambda_1f_1-\lambda_2f_2}\right)_x,
\end{eqnarray}
 where $f_1$ and $f_2$ are two distinct solutions of the Riccati equations \eqref{riccati1}-\eqref{riccati2}.

\subsection*{Reduction}
The eigenvalues $\lambda_k$ have to be real or pairs of complex conjugate values when we choose the reduction $r_{[k]}=q_{[k]}^*$. This reduction condition gives the following relations:

\begin{eqnarray}
 f_kf_k^*&=&1 \hspace{0.2cm}\text{for real} \hspace{0.2cm} \lambda_k,\\
f_m&=&\frac{1}{f_k^*}\hspace{0.2cm}\text{when}\hspace{0.2cm}\lambda_m =\lambda_k^*\hspace{0.2cm}(m\neq k),
\end{eqnarray}
where $f_i$ is a solution of the Riccati equations \eqref{riccati1}-\eqref{riccati2} $(i,k,m\in \mathbb{N})$.

\section{Particular solutions}\label{PS}

\subsection{Solutions for the vacuum}

For $q=r=0$, the Riccati equations \eqref{riccati1}-\eqref{riccati2} transforms into the first-order linear system
\begin{eqnarray}\label{1stLDE}
 f_x-2\lambda^2if&=&0\\
 f_t-4i\lambda^4f&=&0
\end{eqnarray}
which has a solution
\begin{eqnarray}
 f= c e^{2\lambda^2\left(x+2\lambda^2 t\right)i},
\end{eqnarray}
where $c$ is an arbitrary integration constant.

\subsection*{Case 1 ($n=1$)}
For one single Darboux transformation-due to the required reduction $r=q^*$, we have to take $\lambda_1=\lambda$ real and $|c|=1$.
By choosing arbitrary constant $c=1$, we have
\begin{eqnarray}
 f_1=e^{2\lambda_1^2\left(x+2\lambda_1^2 t\right)i}.
\end{eqnarray}
By substituting $f_1$ into \eqref{qr1}, we obtain a new solution $q_{[2]}$ for DNLS equation \eqref{DNLS} as
\begin{eqnarray}
 q_{[2]}=-2i\lambda_1  e^{2\lambda_1^2\left(x+2\lambda_1^2 t\right)i},
\end{eqnarray}
where $r_{[2]}=\left(q_{[2]}\right)^*$.\\
This, of course, is not a soliton but a periodic solution. It is obvious that $\left|q_{[2]}\right|^2=constant$ so it satisfies a linear equation
$iq_t+cq_x+q_{xx}=0$ obtained from \eqref{DNLS}, where $c$ is a constant. Thus, it is not an interesting solution we would like to have by the Darboux transformation.

\subsection*{Case 2 ($n=2$)}
Substituting $f_1=e^{2\lambda_1^2\left(x+2\lambda_1^2 t\right)i}$ and  $f_2=e^{2\lambda_2^2\left(x+2\lambda_2^2 t\right)i}$ into \eqref{q3} and then letting $\lambda_1=\xi+\eta$ and $\lambda_2=\xi-\eta$,
we obtain a new solution $q_{[3]}$ for DNLS equation \eqref{DNLS} as
\begin{eqnarray} q_{[3]}&=&-2i\left(\lambda_1^2-\lambda_2^2\right)\frac{\lambda_1f_1-\lambda_2f_2}{\left(\lambda_2f_1-\lambda_1f_2\right)^2}\\
        &=&-4i\xi\eta~\frac{\left(\eta\cos F+\xi i\sin F\right)^3}{\left(\xi^2+\left(\eta^2-\xi^2\right)\cos^2 F \right)^2}~e^{-iG}\label{sol.ab},
\end{eqnarray}
where $F=4\xi\eta\left(x+4\left(\xi^2+\eta^2\right)t\right)$ and $G=2\left(\xi^2+\eta^2\right)x+4\left(\xi^4+6\xi^2\eta^2+\eta^4\right)t$.

In order that $r_{[3]}=q_{[3]}^*$, $\lambda_1$ and $\lambda_2$ are either real $(\xi,\eta\in \mathbb{R})$ or complex conjugate eigenvalues $(\xi\in \mathbb{R}, \eta\in i\mathbb{R})$.
This solution holds in both cases. For real eigenvalues $\lambda_1,\lambda_2 \in \mathbb{R}$, both $F,G$ are real and for the complex case, $G$ is real while $F=iH$ is purely imaginary. For the complex case, taking $\lambda_{1,2}=\kappa \pm i \tau$, where $\xi=\kappa$ and $\eta=i\tau$, \eqref{sol.ab} gives soliton solution of the DNLS generated by the two-fold Darboux transformation
\begin{eqnarray}
 q_{[3]}=-4i\kappa\tau \frac{\left(\tau \cosh H+i\kappa\sinh H\right)^3}{\left(\kappa^2-\left(\kappa^2+\tau^2\right)\cosh^2 H\right)^2}e^{-iG},
\end{eqnarray}
where $G=2\left(\kappa^2-\tau^2\right)x+4\left(\kappa^4-6\kappa^2\tau^2+\tau^4\right)t$ and $H=4\kappa\tau\left(x+4\left(\kappa^2-\tau^2\right)t\right)$ are real functions.

\subsection{Solutions for non-zero seeds}
For $q,r\neq 0$ and $r=q^*$, we can easily find a periodic solution
\begin{eqnarray}\label{perq}
 q=ke^{ia\left[x-(a-k^2)t\right]}
\end{eqnarray}
of the DNLS equation \eqref{DNLS}, where $a$ and $k$ are real numbers. We use this as the seed solution for application of Darboux transformations.

We already know that the solution of \eqref{DNLS} is given in terms of solutions $f_i$ of Riccati system \eqref{riccati1}-\eqref{riccati2}. For given $q$ in \eqref{perq}, we will solve the Riccati system for the function $f(x,t)$. If we define
\begin{eqnarray}
 f=\mu(x,t)e^{-i\alpha},
\end{eqnarray}
where $\alpha=a\left[x-(a-k^2)t\right]$, as a solution of the Riccati system \eqref{riccati1}-\eqref{riccati2} for given $q$ in \eqref{perq},
we end up with the single Riccati equation with constant coefficients
\begin{eqnarray}\label{riccatiC}
 \mu_x+\lambda k \mu^2 -i\left(a+2\lambda^2\right)\mu-\lambda k=0,
\end{eqnarray}
where $\mu_t=\left(k^2-a+2\lambda^2\right)\mu_x$. By letting \begin{eqnarray}\mu=\frac{u_x}{\lambda ku},\end{eqnarray}
the Riccati equation \eqref{riccatiC} transforms into the second-order linear partial differential equation with constant coefficients
\begin{eqnarray}\label{LPDE}
 u_{xx}-i\left(a+2\lambda^2\right)u_x-\lambda^2 k^2 u=0,
\end{eqnarray}
where $u_t=\left(k^2-a+2\lambda^2\right)u_x$. Solving this equation and then substituting $u$ into \eqref{riccatiC},
we obtain the general solution of the Riccati equation as
\begin{eqnarray}\label{mu}
\mu(x,t)=\frac{1}{2\lambda k}\left(Ai+D\frac{c_1e^{\frac{1}{2} D(x+Bt)}-c_2e^{-\frac{1}{2} D(x+Bt)}}{c_1e^{\frac{1}{2} D(x+Bt)}+c_2e^{-\frac{1}{2} D(x+Bt)}}\right),
\end{eqnarray}
where
\begin{eqnarray}
 A=a+2\lambda^2,~~B=k^2-a+2\lambda^2,~~D=\sqrt{4\lambda^2k^2-A^2},
\end{eqnarray}
and $c_1,~c_2$ integration constants, obtained from \eqref{LPDE}.

\subsection*{Case 3 ($n=1$)}
For single real eigenvalue $\lambda_1$, substituting $r=ke^{-i\alpha}$ and $f_1=\mu_1(x,t)e^{-i\alpha}$ into \eqref{qr1} gives the solution
\begin{eqnarray}
  q_{[2]}=\mu_1 \left(k\mu_1-2i\lambda_1\right)e^{-i\alpha},
\end{eqnarray}
where $\alpha(x,t)=a\left[x-a\left(a-k^2\right)t\right]$ and
$\mu_1=\mu$ is given in \eqref{mu} with relabeled coefficients $A_1=A,~ B_1=B,~ D_1=D$ such that $|\mu_1|=1$ and $A_1=a+2\lambda_1^2$, $B_1=k^2-a+2\lambda_1^2$, $D_1=\sqrt{4k^2\lambda_1^2-A_1^2}~$.
It is such that
\begin{eqnarray}\label{q2ps}
 \left|q_{[2]}\right|^2=\left|k\mu_1-2i\lambda_1\right|^2.
\end{eqnarray}
We show below that $D_1^2<0$ and $D_1^2>0$ produce the periodic and soliton solutions respectively.
\subsubsection*{Periodic solution} For $D_1^2=4k^2\lambda_1^2-\left(a+2\lambda_1^2\right)^2<0$, \eqref{mu} gives us
\begin{eqnarray}
 \mu_1(x,t)=i\left(\frac{ke^{\frac{1}{2}i\beta_1}+2\lambda_1e^{-\frac{1}{2}i\beta_1}}{2\lambda_1e^{\frac{1}{2}i\beta_1}+ke^{-\frac{1}{2}i\beta_1}}\right)
\end{eqnarray}
in which have chosen $k^2=2a$, where $\beta_1(x,t)=\left(a-2\lambda_1^2\right)\left[x+\left(a+2\lambda_1^2\right)t\right].$ It can be easily seen that the relation $\mu_1\mu_1^*=1$ holds. By substituting $\mu_1$ into \eqref{q2ps}, we obtain
\begin{eqnarray}
  \left|q_{[2]}\right|^2=\frac{2\left(a-2\lambda_1^2\right)}{a+2\lambda_1^2+2k\lambda_1\cos \beta_1},
\end{eqnarray}
which is a periodic solution.

\subsubsection*{Soliton solution} For $D_1^2=4k^2\lambda_1^2-\left(a+2\lambda_1^2\right)^2>0$, \eqref{q2ps} gives soliton solution as
\begin{eqnarray}
 \left|q_{[2]}\right|^2=k^2-2a-\frac{D_1^2}{(a+2\lambda_1^2)/2+\delta \lambda_1 k\cosh{\gamma_1}},
\end{eqnarray}
where $\gamma_1(x,t)=D_1\left[x+\left(k^2-2a+2\lambda_1^2\right)t\right]$ and $\delta=\pm1$.

\subsection*{Case 4 ($n=2$)} In this case, we have two eigenvalues $\lambda_1$ and $\lambda_2$. For solutions such that $r=q^*$, these eigenvalues are either real or complex conjugate to each other and satisfy the relations $|f_1|=|f_2|=1$  or $f_1f_2^*=1$ respectively, where $f_1$ and $f_2$ are two distinct solutions for the Riccati system \eqref{riccati1}-\eqref{riccati2}. By substituting $q=ke^{i\alpha}$ and $f_1=\mu_1(x,t)e^{-i\alpha}$, $f_2=\mu_2(x,t)e^{-i\alpha}$ into \eqref{q3}, we have the following solution
\begin{eqnarray}\label{q3.p}
  q_{[3]}=\left(\frac{kF_2-2i\Lambda}{F_1^2}\right)F_2 e^{i\alpha},
\end{eqnarray}
where $\alpha(x,t)=a\left[x-\left(a-k^2\right)t\right]$, $\Lambda=\lambda_1^2-\lambda_2^2$, $F_1(x,t)=\lambda_2\mu_1-\lambda_1\mu_2$, $F_2(x,t)=\lambda_1\mu_1-\lambda_2\mu_2$. Here $\mu_1(x,t)$ and $\mu_2(x,t)$ are two distinct solutions, given by \eqref{mu}, for the Riccati equation with constant coefficients \eqref{riccatiC}.  The functions $\mu_1, \mu_2$ with the eigenvalues $\lambda_1, \lambda_2$ either hold $(R1)$ $\mu_1\mu_1^*=\mu_2\mu_2^*=1$
for $\lambda_1,\lambda_2\in\mathbb{R}$ and so $\Lambda$ is real or $(R2)$ $\mu_1\mu_2^*=1$ for $\lambda_2=\lambda_1^*$ and so $\Lambda$ is pure imaginary. The solution above can be rewritten as
\begin{eqnarray}\label{q3.pa}
 \left|q_{[3]}\right|^2=\left|k+2\frac{\Lambda}{\Omega_1}\right|^2,
\end{eqnarray}
where $\Omega_1=iF_2$ and $\Omega_2=iF_1$. This holds for both $(R1)$ and $(R2)$. This result is consistent with \cite{Steudel}. This can be rewritten as
\begin{eqnarray} \left|q_{[3]}\right|^2=k^2\pm2\frac{\Lambda}{\left|\Omega_1\right|^2}\left[2\Lambda+k\left(\Omega_1\pm\Omega_1^*\right)\right],
\end{eqnarray}
which holds for $(R1)$ and $(R2)$ respectively. An example for $(R1)$ is given below.

\subsubsection*{Periodic solution}  We have periodic solution for the choice $(R1)$ with $k^2=2a$ as
 \begin{eqnarray}
 \left|q_{[3]}\right|^2=k^2+\frac{m_0+m_1 \cos\beta_1+m_2 \cos\beta_2}{n_0+n_1\cos\beta_1+n_2\cos\beta_2+n_3\cos\beta_3},
\end{eqnarray}
where
\begin{eqnarray*}
 \beta_1&=&\left(a-2\lambda_1^2\right)\left[x+\left(a+2\lambda_1^2\right)t\right],\\
 \beta_2&=&\left(a-2\lambda_2^2\right)\left[x+\left(a+2\lambda_2^2\right)t\right],\\
 \beta_3&=&2\Lambda\left[x+2\left(\lambda_1^2+\lambda_2^2\right)t\right],
\end{eqnarray*}
and
\begin{eqnarray*}
 m_0&=&4\Lambda\left[\left(a-2\lambda_1^2\right)^2\left(a+2\lambda_2^2\right)-\left(a+2\lambda_1^2\right)\left(a-2\lambda_2^2\right)^2\right],\\
 m_1&=&-8k\lambda_1\Lambda\left(a-2\lambda_2^2\right)^2,\\
 m_2&=&8k\lambda_2\Lambda\left(a-2\lambda_1^2\right)^2,\\
 n_0&=&\left(a-2\lambda_1^2\right)^2\left(a+2\lambda_2^2\right)+\left(a+2\lambda_1^2\right)\left(a-2\lambda_2^2\right)^2
 -k^2\left(a-2\lambda_1^2\right)\left(a-2\lambda_2^2\right),\\
 n_1&=&4k\lambda_1\Lambda\left(a-2\lambda_2^2\right),\\
 n_2&=&-4k\lambda_2\Lambda\left(a-2\lambda_1^2\right),\\
 n_3&=&-4\lambda_1\lambda_2\left(a-2\lambda_1^2\right)\left(a-2\lambda_2^2\right).
\end{eqnarray*}
In the second case ($R2$), a similar result is obtained expressed in terms of sines-cosines and hyperbolic sines-cosines.

\section{Conclusion}

In this paper, we have established and proved a new theorem for Darboux transformation of operators with no derivative term.
This has a similar structure to the theorem on the standard Darboux transformation for general operators.
We have constructed solutions in quasideterminant forms for the DNLS equation.
These quasideterminants are expressed in terms of $f_i$ functions, where $f_i (i\in\mathbb{N})$ are solutions of Riccati systems.
It should be pointed out that these solutions are derived from linear eigenvalue problems $L(\Phi)=M(\Phi)=0$, where $\Phi=(\phi,\psi)^T$ and the linear operators $L$, $M$ are given in \eqref{LaxL}-\eqref{LaxM}. By letting $f_i=\psi_i / \phi_i$ in \eqref{qodd}-\eqref{oddmatrix} and \eqref{qeven}-\eqref{evenmatrix},
we easily write down the quasideterminant solutions in terms of solutions $\phi_i , \psi_i$ of linear partial differential equations (eigenvalue problems) as given in \cite{Xu}.

It should be emphasised that these solutions arise naturally from the Darboux transformation we present here.
Our theorem provides a natural and universal approach for operators with no derivative term.
Furthermore, for the DNLS equation, parametric and soliton solutions for zero and non-zero seeds have been presented here.
Finally, it is important to point out that our approach can be applied to other integrable systems in which their Lax operators have no derivative term.

\subsection*{Acknowledgments} The author (H.Y.) wishes to express his thanks to the School of Mathematics and Statistics, University of Glasgow, for hosting him as a Honorary Research Fellow.

\end{document}